\documentclass[letter, 10pt, conference]{ieeeconf}      

\IEEEoverridecommandlockouts                              

\usepackage{amsmath}
\usepackage{amssymb}
\usepackage{amsfonts}
\usepackage{graphicx}
\usepackage{enumerate}
\usepackage{mathrsfs}
\usepackage{float}
\usepackage{cite}
\usepackage{tikz}
\usetikzlibrary{matrix,arrows,decorations.pathmorphing}
\usepackage{verbatim}
\usepackage{mathtools}
\usepackage{caption}
\usepackage{subcaption}
\usepackage{soul}
\setstcolor{magenta}
\usetikzlibrary{arrows}
\usepackage{tabularx}

\graphicspath{{./figs/}}

\title{\LARGE \bf
Linear Feedback Controller for Homogeneous Polynomial Systems$^{*}$}

\author{Shaoxuan Cui$^{1}$, Qi Zhao$^{2}$, Guanlin Li$^{2}$, Hildeberto Jardón-Kojakhmetov$^{1}$ and Ming Cao$^{3}$ 
\thanks{$^{1}$ S. Cui, and H. Jard\'on-Kojakhmetov are with the Bernoulli Institute for Mathematics, Computer Science and Artificial Intelligence, University of Groningen, Groningen, 9747 AG Netherlands {\tt\small \{s.cui, h.jardon.kojakhmetov\}@rug.nl}}
\thanks{$^{2}$ Q. Zhao and G. Li are with the School of Data Science, Qingdao University of Science and Technology, Qingdao, 266061 China. (e-mail: zhaoqi\_1995@163.com; liguanlin0619@gmail.com) }
\thanks{$^{3}$ M. Cao is with the Engineering and Technology institute Groningen, University of Groningen, Groningen, 9747 AG Netherlands {\tt\small m.cao@rug.nl}}
\thanks{$^{*}$ S. Cui and M. Cao were supported in part by the Netherlands Organization for Scientific Research (NWO-Vici-19902). Q. Zhao was supported by the Qingdao Natural Science Foundation (Grant No. 25-1-1-121-zyyd-jch). }
}

\begin{document}

\maketitle

\thispagestyle{empty}
\pagestyle{empty}

\newtheorem{corollary}{Corollary}
\newtheorem{remark}{Remark}
\newtheorem{lemma}{Lemma}
\newtheorem{theorem}{Theorem}
\newtheorem{example}{Example}
\newtheorem{definition}{Definition}
\newtheorem{proposition}{Proposition}

\begin{abstract}
This paper studies stabilization and its corresponding closed-loop region-of-attraction (ROA) for homogeneous polynomial dynamical systems whose nonlinear term admits an orthogonally decomposable (ODECO) tensor representation. While recent tensor-based results provide explicit solutions and sharp global characterizations for open-loop ODECO systems, closed-loop synthesis and computable ROA estimates are still often dominated by local linearization or Lyapunov/SOS (sum of squares) methods, which can be conservative and computationally demanding. We propose a \emph{structure-preserving linear feedback} design that shares the ODECO eigenbasis of the system's tensor, thereby enabling closed-form trajectory expressions, explicit convergence/escape thresholds, and sharp ROA characterizations. Under mild conditions, we further derive robustness/ISS-type bounds for bounded disturbances. Numerical examples validate the theoretical results.
\end{abstract}

\section{Introduction}
\label{sec:intro}

Polynomial dynamical systems (PDS) appear in many nonlinear modeling paradigms,
including Lotka--Volterra dynamics \cite{goh1976global,cui2025analysis}, chemical reaction and biochemical kinetics models \cite{angeli2009tutorial},
as well as polynomial epidemic-type compartmental systems \cite{cui2025sis,cui2022discrete}.
Among them, homogeneous polynomial dynamical systems (HPDS) form a basic, relatively simple but widely studied subclass:
they are mathematically simpler, yet still capture essential nonlinear features such as finite-time escape  
and nontrivial regions of attraction (ROA) \cite{ahmadi2019algebraic,samardzija1983stability,chen2021stability,chen2022explicit,baillieul1980geometry}.
Because of their intrinsic nonlinearity, controller synthesis and stability certification for HPDS are still
dominated by local linearization-based designs \cite{khalil2002nonlinear,isidori1995nonlinear,slotine1991applied} or Lyapunov/SOS(sum of squares)-type methods \cite{saat2017analysis,papachristodoulou2005tutorial,parrilo2000structured}.
This often results in limited operating ranges, high computational costs, and conservative ROA estimates.

A recent line of research shows that \emph{tensor algebra} can endow HPDS with an explicit form of solution when the associated  tensor admits special decompositions.
In particular, Chen \cite{chen2022explicit} establishes explicit solution formulas and sharp stability characterizations for continuous-time \emph{orthogonally decomposable} (ODECO) HPDS by exploiting Z-eigenvalues and the orthogonal tensor decomposition. Then, Chen \cite{chen2021stability} further
derives explicit solution formulas for discrete-time counterpart using the same tensor toolbox.

In this paper, we develop a \emph{structure-preserving linear feedback} framework for homogeneous ODECO polynomial systems.
Our methodology is based on tensor orthogonal decomposition: we enforce that the linear feedback shares the ODECO basis of the dynamic tensor, so that the closed-loop dynamics remain exactly modal-decoupled into independent scalar polynomial ODEs. 
This enables explicit solutions, explicit blow-up/convergence thresholds, and a sharp estimate of the region of attraction (ROA).

The main contributions are: (i) a \emph{tunable} shared-basis linear feedback controller (via modal gains ${\kappa_r}$) that preserves the ODECO structure and yields exact decoupling; (ii) closed-form modal trajectories that provide sharp, explicitly \emph{shaped} ROA thresholds and convergence/escape time estimates; and (iii) for even degree under matched bounded disturbances, robust invariant sets and ISS-type ultimate bounds.

\section{Preliminaries}\label{sec:preliminaries}

Let $\mathbb{R}$ be the set of real numbers and $\|\cdot\|_2$ the Euclidean norm.
For $n\in\mathbb{N}$, denote $[n]:=\{1,\dots,n\}$.

A $k$th-order (cubical) tensor is an array $\mathcal{A}\in\mathbb{R}^{n\times\cdots\times n}$ with entries
$a_{i_1\cdots i_k}$, $i_j\in[n]$.
The tensor $\mathcal{A}$ is \emph{supersymmetric} if $a_{i_{\pi(1)}\cdots i_{\pi(k)}}=a_{i_1\cdots i_k}$ for any permutation $\pi$.

For $v\in\mathbb{R}^n$, the tensor $v^{\otimes k}\in\mathbb{R}^{n\times\cdots\times n}$ has entries
$(v^{\otimes k})_{i_1\cdots i_k}:=v_{i_1}\cdots v_{i_k}$.
Given $\mathcal{A}\in\mathbb{R}^{n\times\cdots\times n}$ and $x\in\mathbb{R}^n$, define the standard contraction
$\mathcal{A}x^{k-1}\in\mathbb{R}^n$ by
\begin{equation}\label{eq:prelim_tensor_contraction}
(\mathcal{A}x^{k-1})_i := \sum_{i_2,\dots,i_k=1}^n a_{i\,i_2\cdots i_k}\,x_{i_2}\cdots x_{i_k},\qquad i\in[n].
\end{equation}
If $\mathcal{A}$ is supersymmetric, the associated homogeneous polynomial (scalar form) is
\begin{equation}\label{eq:prelim_hom_poly}
\mathcal{A}x^{k} := \sum_{i_1,\dots,i_k=1}^n a_{i_1\cdots i_k}\,x_{i_1}\cdots x_{i_k}.
\end{equation}
In particular, homogeneous polynomial differential equations of degree $k-1$ can be written compactly as $\dot x=\mathcal{A}x^{k-1}$.

For a real supersymmetric tensor $\mathcal{A}\in\mathbb{R}^{n\times\cdots\times n}$, a pair $(\lambda,v)$ with
$\lambda\in\mathbb{R}$ and $v\in\mathbb{R}^n\setminus\{0\}$ is a \emph{Z-eigenpair} if
\begin{equation}\label{eq:prelim_Zeig}
\mathcal{A}v^{k-1}=\lambda v,\qquad v^\top v=1.
\end{equation}

A supersymmetric tensor $\mathcal{A}$ is \emph{orthogonally decomposable (ODECO)} if it admits a decomposition of the form
$\mathcal{A}=\sum_{r=1}^n \lambda_r\, v_r^{\otimes k},$
where $\{v_r\}_{r=1}^n$ is an orthonormal basis of $\mathbb{R}^n$ and $\lambda_r\in\mathbb{R}$.
In this case, $(\lambda_r,v_r)$ are Z-eigenpairs of $\mathcal{A}$.

\section{Homogeneous ODECO System with Linear Feedback}
\subsection{Closed-form solution for the controlled homogeneous ODECO system}\label{sec:closed_form_homogeneous}

We now focus on a homogeneous tensor system stabilized by a linear feedback that preserves the shared ODECO basis:
\begin{equation}\label{eq:hom_odeco_closedloop}
\dot x = u(x) + \mathcal A x^{k-1},\quad u(x)=Kx,\quad k\ge 3,\ \ x(t)\in\mathbb R^n,
\end{equation}
where
\begin{equation}\label{eq:hom_odeco_decomp}
\mathcal A=\sum_{r=1}^{n}\lambda_r\, v_r^{\otimes k},
\qquad
K=\sum_{r=1}^{n}\kappa_r\, v_r v_r^\top,
\qquad r=1,\dots,n,
\end{equation}
with an orthonormal basis $\{v_r\}_{r=1}^n$.

\begin{remark}\label{rem:K_not_restrictive}
At first sight, requiring the linear term $K$ to share the same basis $\{v_r\}$ as the ODECO tensor (cf.~\eqref{eq:hom_odeco_decomp}) may appear restrictive.
However, in our setting, we assume the original uncontrolled system is an ODECO homogeneous polynomial system \cite{chen2022explicit}, and the controller is a linear feedback. The controller matrix $K$ is a \emph{design variable} (full-state linear feedback), and \eqref{eq:hom_odeco_decomp} should be interpreted as a \emph{structure-preserving controller choice} rather than a plant constraint.
Indeed, once the ODECO decomposition of the open loop provides the modal directions $\{v_r\}$, selecting
$K=\sum_{r=1}^{n}\kappa_r\, v_r v_r^\top$
amounts to independently tuning the modal linear gains $\{\kappa_r\}$, which is precisely the degree of freedom needed for stabilization and performance shaping in the decoupled coordinates.
This choice enforces that $K$ shares the same eigenvectors as $\mathcal A$, thereby preserving the exact modal decoupling and enabling closed-form trajectory and region-of-attraction characterizations.
\end{remark}

\begin{theorem}\label{thm:closed_form}
Consider~\eqref{eq:hom_odeco_closedloop}--\eqref{eq:hom_odeco_decomp} and define the modal coordinates
$y(t):=V^\top x(t)$ with $V=[v_1,\dots,v_n]$ so that $y_r(t)=v_r^\top x(t)$.
Then the dynamics decouple as
\begin{equation}\label{eq:mode_bernoulli}
\dot y_r = \kappa_r y_r + \lambda_r y_r^{k-1},
\qquad r=1,\dots,n,
\end{equation}
and the state is reconstructed by
\begin{equation}\label{eq:reconstruct_x_closed}
x(t)=\sum_{r=1}^{n} y_r(t)\, v_r.
\end{equation}
Let $p:=k-2>0$ and denote $y_{r,0}:=y_r(0)=v_r^\top x(0)$.

\smallskip
\noindent
(i) If $\kappa_r\neq 0$, then for any time interval on which $y_r(t)\neq 0$ (hence its sign is constant), the trajectory admits the closed-form representation
\begin{equation}\label{eq:closed_form_kappa_nonzero}
y_r(t)
=
y_{r,0}\Bigg[
e^{-p\kappa_r t}
-\frac{\lambda_r}{\kappa_r}\,y_{r,0}^{p}\Big(1-e^{-p\kappa_r t}\Big)
\Bigg]^{-1/p},
\end{equation}
valid on the maximal interval where the bracketed term is positive (so that the real-valued power is well-defined).

\smallskip
\noindent
(ii) If $\kappa_r=0$, then for any time interval on which $y_r(t)\neq 0$,
\begin{equation}\label{eq:closed_form_kappa_zero}
y_r(t)
=
y_{r,0}\,
\Big(1-p\,\lambda_r\, y_{r,0}^{p}\, t\Big)^{-1/p},
\end{equation}
valid on the maximal interval where $1-p\,\lambda_r\, y_{r,0}^{p}\, t>0$.
\end{theorem}

\begin{proof}
Write $x=Vy=\sum_{r=1}^n y_r v_r$ with $y=V^\top x$ and $y_r=v_r^\top x$.
Using~\eqref{eq:hom_odeco_decomp} and orthonormality, we have
$Kx
=\sum_{r=1}^{n}\kappa_r\, y_r v_r,
\qquad
\mathcal A x^{k-1}
=\sum_{r=1}^{n}\lambda_r\, (v_r^\top x)^{k-1} v_r
=\sum_{r=1}^{n}\lambda_r\, y_r^{k-1} v_r.$
Substituting into~\eqref{eq:hom_odeco_closedloop} gives
$\dot x=\sum_{r=1}^{n}\big(\kappa_r y_r+\lambda_r y_r^{k-1}\big)v_r.$
Projecting onto each $v_r$ yields~\eqref{eq:mode_bernoulli}, and the reconstruction follows from $x=Vy$, i.e.,~\eqref{eq:reconstruct_x_closed}.
Fix a mode $r$ and consider~\eqref{eq:mode_bernoulli}.
Let $p=k-2$ and work on any interval where $y_r(t)\neq 0$.
Define $u_r(t):=y_r(t)^{-p}$. Then
$\dot u_r
=
-p\,y_r^{-p-1}\dot y_r
=
-p\big(\kappa_r y_r^{-p}+\lambda_r\big)
=
-p\kappa_r u_r - p\lambda_r.$
If $\kappa_r\neq 0$, solving this linear ODE gives
$u_r(t)
=
e^{-p\kappa_r t}u_r(0) - \frac{\lambda_r}{\kappa_r}\Big(1-e^{-p\kappa_r t}\Big),
\qquad u_r(0)=y_{r,0}^{-p}.$
Equivalently,
$y_r(t)^{-p}
=
y_{r,0}^{-p}
\Bigg[
e^{-p\kappa_r t}
-\frac{\lambda_r}{\kappa_r}\,y_{r,0}^{p}\Big(1-e^{-p\kappa_r t}\Big)
\Bigg].$
Inverting yields~\eqref{eq:closed_form_kappa_nonzero}. The prefactor $y_{r,0}$ fixes the real-valued branch so that $y_r(t)$ preserves the sign of $y_{r,0}$ on the considered interval.
If $\kappa_r=0$,~\eqref{eq:mode_bernoulli} becomes $\dot y_r=\lambda_r y_r^{p+1}$, and integrating gives~\eqref{eq:closed_form_kappa_zero}.
Finally,~\eqref{eq:reconstruct_x_closed} follows from $x=Vy=\sum_r y_r v_r$.
\end{proof}

\begin{remark}
A key advantage of the proposed controller is its \emph{direct tunability}: by selecting the modal gains $\{\kappa_r\}$, one can explicitly shape the ROA thresholds (mode-wise), adjust convergence and escape-time profiles via closed-form formulas, and obtain disturbance-to-state guarantees through the derived ISS-type bounds. These advantages will be characterized explicitly in the following Sections.
\end{remark}

\subsection{Computable region of attraction}\label{sec:roa}

In this subsection we characterize a computable region of attraction (ROA) for the origin of the controlled homogeneous ODECO system~\eqref{eq:hom_odeco_closedloop}--\eqref{eq:hom_odeco_decomp}.
Throughout, let $p:=k-2>0$ and denote $y_{r,0}:=v_r^\top x(0)$.
We emphasize that the parity of $p$ (equivalently, of $k$) affects the geometry of the ROA in the modal coordinates.

\begin{proposition}\label{prop:local_exp}
Consider~\eqref{eq:hom_odeco_closedloop}--\eqref{eq:hom_odeco_decomp}.
If $\kappa_r<0$ for all $r=1,\dots,n$, then the origin is a locally exponentially stable equilibrium of the closed-loop system.
\end{proposition}

\begin{proof}
Under~\eqref{eq:hom_odeco_decomp}, each mode satisfies the scalar ODE~\eqref{eq:mode_bernoulli}.
The linearization at the origin is $\dot y_r=\kappa_r y_r$.
If $\kappa_r<0$ for all $r$, then each decoupled scalar subsystem  is locally exponentially stable, hence so is the original system.
\end{proof}

We next provide an explicit ROA characterization that guarantees convergence to the origin, even in the presence of ``destabilizing'' higher-order modes.

\begin{theorem}\label{thm:roa_box}
Consider~\eqref{eq:hom_odeco_closedloop}--\eqref{eq:hom_odeco_decomp} and assume $\kappa_r<0$ for all $r$.
Define the index sets
$\mathcal I_+:=\{r:\lambda_r>0\},\qquad
\mathcal I_-:=\{r:\lambda_r<0\}.$
For each $r\in\mathcal I_+$, define the (mode-wise) nonzero equilibrium threshold
$c_{r}:=\Bigl(-\frac{\kappa_r}{\lambda_r}\Bigr)^{1/p}>0.$

\smallskip
\noindent
\textnormal{(i) Even $p$ (equivalently, even $k$).}
Define
$\mathcal R_{\mathrm{even}}
:=
\Bigl\{x_0\in\mathbb R^n:\ 
|v_r^\top x_0|<c_{r}\ \ \forall r\in\mathcal I_+\Bigr\}.$
Then, for every $x(0)=x_0\in\mathcal R_{\mathrm{even}}$, the corresponding solution of~\eqref{eq:hom_odeco_closedloop}
exists for all $t\ge 0$, satisfies $\lim_{t\to\infty} x(t)=0$, and $\mathcal R_{\mathrm{even}}$ is forward invariant.

\smallskip
\noindent
\textnormal{(ii) Odd $p$ (equivalently, odd $k$).}
Define
$\mathcal R_{\mathrm{odd}}
:=
\Bigl\{x_0\in\mathbb R^n:\ 
v_r^\top x_0<c_{r}\ \ \forall r\in\mathcal I_+,\ \ 
v_r^\top x_0>c_{r}\ \ \forall r\in\mathcal I_-\Bigr\}.$
Then, for every $x(0)=x_0\in\mathcal R_{\mathrm{odd}}$, the corresponding solution of~\eqref{eq:hom_odeco_closedloop}
exists for all $t\ge 0$, satisfies $\lim_{t\to\infty} x(t)=0$, and $\mathcal R_{\mathrm{odd}}$ is forward invariant.
\end{theorem}

\begin{proof}
By Theorem~\ref{thm:closed_form}, the closed-loop dynamics decouple into scalar modes
$\dot y_r=\kappa_r y_r+\lambda_r y_r^{p+1},\qquad r=1,\dots,n,$
and $x(t)=Vy(t)=\sum_{r=1}^n y_r(t)v_r$.
Fix a mode $r$ and define
$g_r(s):=\kappa_r+\lambda_r s^{p},\qquad f_r(s):=s\,g_r(s)=\kappa_r s+\lambda_r s^{p+1}.$
Note that $y_r=0$ is always an equilibrium, and any nonzero equilibrium satisfies $g_r(y_r)=0$.

\smallskip
\noindent
\emph{Even $p$.}
When $p$ is even, $y^p=|y|^p\ge 0$. Hence for any $\lambda_r\le 0$ we have
$g_r(|y|)=\kappa_r+\lambda_r|y|^p\le \kappa_r<0\quad \forall\,y\in\mathbb R,$
and therefore $\dot y_r = y_r\,g_r(|y_r|)$ has the sign opposite to $y_r$, implying $y_r(t)\to 0$ globally.

For $r\in\mathcal I_+$, $c_r>0$ solves $g_r(c_r)=0$.
If $|y_r|<c_r$, then $g_r(|y_r|)<0$, hence $\dot y_r$ has the sign opposite to $y_r$ and the interval $(-c_r,c_r)$ is forward invariant and attracts to $0$.
Thus, if $x_0\in\mathcal R_{\mathrm{even}}$, all modes converge to $0$ and $x(t)=Vy(t)\to 0$.
Forward invariance follows since each constraint $|y_r|<c_r$ is preserved over time.

\smallskip
\noindent
\emph{Odd $p$.}
When $p$ is odd, $s\mapsto s^p$ is strictly increasing on $\mathbb R$.
We treat three cases.

\emph{Case 1: $r\in\mathcal I_+$ (i.e., $\lambda_r>0$).}
Then $c_r>0$ and $g_r(c_r)=0$, with $g_r(s)<0$ for all $s<c_r$.
Hence, if $y_r(0)<c_r$, we have $g_r(y_r(t))<0$ as long as $y_r(t)<c_r$, and thus
$\dot y_r = y_r g_r(y_r)
\ \text{has sign opposite to } y_r \text{ when } y_r>0,\ \text{and is positive when } y_r<0.$
Consequently, the half-line $(-\infty,c_r)$ is forward invariant and $y_r(t)\to 0$.

\emph{Case 2: $r\in\mathcal I_-$ (i.e., $\lambda_r<0$).}
Then $c_r<0$ and $g_r(c_r)=0$.
Because $g_r$ is increasing, we have $g_r(s)<0$ for $s\in(c_r,0)$ and $g_r(s)>0$ for $s<c_r$.
Thus, if $y_r(0)>c_r$, then $y_r(t)$ cannot cross $c_r$:
on $(c_r,0)$ we have $y_r<0$ and $g_r(y_r)<0$, so $\dot y_r=y_rg_r(y_r)>0$ and the trajectory moves upward toward $0$;
on $(0,\infty)$ we have $y_r>0$ and $g_r(y_r)\le \kappa_r<0$ (since $\lambda_r<0$ and $y_r^p>0$), so $\dot y_r<0$ and the trajectory decreases toward $0$.
Therefore, the half-line $(c_r,\infty)$ is forward invariant and attracts to $0$.

\emph{Case 3: $\lambda_r=0$.}
Then $\dot y_r=\kappa_r y_r$ and $y_r(t)\to 0$ for all initial conditions.

Combining the cases, if $x_0\in\mathcal R_{\mathrm{odd}}$, then all modes remain in their respective invariant half-lines and converge to $0$, hence $x(t)=Vy(t)\to 0$.
Forward invariance follows since each modal inequality is preserved over time.
\end{proof}

\begin{remark}[Special cases and conservativeness]\label{rem:roa_special}
If $\mathcal I_+=\emptyset$ (i.e., $\lambda_r\le 0$ for all $r$), then:
for even $p$ we obtain global convergence to the origin for any $x_0\in\mathbb R^n$ under $\kappa_r<0$;
for odd $p$, global convergence still holds when $\mathcal I_-=\emptyset$ (i.e., $\lambda_r=0$ for all $r$), but if $\mathcal I_-\neq\emptyset$ then the additional lower-bound constraints $v_r^\top x_0>c_r$ for $r\in\mathcal I_-$ are necessary to exclude the blow-down side $(-\infty,c_r)$ along those modes.
When $\mathcal I_+\neq\emptyset$, $\mathcal R_{\mathrm{even}}$ and $\mathcal R_{\mathrm{odd}}$ provide explicit ROA.
Moreover, for even $p$ the threshold is symmetric in each destabilizing mode, whereas for odd $p$ the ROA is described by modal inequalities determined by the sign of $\lambda_r$.
\end{remark}

\begin{remark}[Boundary cases]\label{rem:blowup_boundary}
Fix a mode with $\kappa_r<0$ and $\lambda_r\neq 0$, and let $c_r$ so that $g_r(s):=\kappa_r+\lambda_r s^p$ satisfies $g_r(c_r)=0$.
\textnormal{(i) Even $p$:}
For $\lambda_r>0$, if $|y_{r,0}|=c_r$ then $y_r(t)\equiv \pm c_r$ is an equilibrium and the trajectory does not converge to the origin.
Thus the modal ROA threshold $|y_{r,0}|<c_r$ is sharp for each destabilizing mode.
If $\lambda_r<0$, no finite-time escape occurs and the mode is globally attracted to $0$ under $\kappa_r<0$.
\textnormal{(ii) Odd $p$:}
For $\lambda_r>0$, the nonzero equilibrium is $y=c_r>0$; thus $y_{r,0}=c_r$ yields the constant solution $y_r(t)\equiv c_r$, while any $y_{r,0}>c_r$ escapes to $+\infty$.
If $y_{r,0}<c_r$, then the mode is forward complete and converges to $0$ (in particular, any $y_{r,0}<0$ converges to $0$).
For $\lambda_r<0$, the nonzero equilibrium is $y=c_r<0$; thus $y_{r,0}=c_r$ yields $y_r(t)\equiv c_r$, any $y_{r,0}<c_r$ escapes to $-\infty$, and any $y_{r,0}>c_r$ is forward complete and converges to $0$.
\end{remark}

\subsection{Convergence-time estimate within the ROA}\label{sec:roa_time}

We quantify the convergence time to a prescribed accuracy within the ROA.
For a given tolerance $\varepsilon>0$, define the (mode-wise) hitting time
$T_{\varepsilon,r}(y_{r,0})
:=
\inf\{t\ge 0:\ |y_r(t)|\le \varepsilon\},$
and the aggregate time
$T_\varepsilon(x_0):=\max_{r=1,\dots,n} T_{\varepsilon,r}(v_r^\top x_0).$
Since $V$ is orthogonal, the bound $|y_r(t)|\le \varepsilon$ for all $r$ implies
$\|x(t)\|_2=\|y(t)\|_2\le \sqrt{n}\,\varepsilon$.

\begin{proposition}\label{prop:settling_time}
Consider~\eqref{eq:hom_odeco_closedloop}--\eqref{eq:hom_odeco_decomp} with $\kappa_r<0$ for all $r$ and let $p:=k-2$.
Fix $\varepsilon>0$ and an initial condition $x_0$ in the ROA
($x_0\in\mathcal R_{\mathrm{even}}$ when $p$ is even, or $x_0\in\mathcal R_{\mathrm{odd}}$ when $p$ is odd; cf.~Theorem~\ref{thm:roa_box}).
Then each mode is well-defined for all $t\ge 0$ and reaches the band $|y_r(t)|\le \varepsilon$ in finite time.
For each $r$, let $y_{r,0}:=v_r^\top x_0$ and $\alpha_r:=|\kappa_r|=-\kappa_r>0$.
If $|y_{r,0}|\le \varepsilon$, then $T_{\varepsilon,r}=0$.
Otherwise ($|y_{r,0}|>\varepsilon$):

\smallskip
\noindent
\textnormal{(i) If $\lambda_r=0$, then}
\begin{equation}\label{eq:Teps_linear}
T_{\varepsilon,r}
=
\frac{1}{\alpha_r}\ln\!\Big(\frac{|y_{r,0}|}{\varepsilon}\Big).
\end{equation}

\smallskip
\noindent
\textnormal{(ii) If $\lambda_r\neq 0$ and $p$ is even, then}
\begin{equation}\label{eq:Teps_even}
T_{\varepsilon,r}
=
\frac{1}{p\alpha_r}
\ln\!\Bigg(
\frac{\varepsilon^{-p}-\lambda_r/\alpha_r}{|y_{r,0}|^{-p}-\lambda_r/\alpha_r}
\Bigg),
\end{equation}
where the logarithm argument is $>1$ for any $x_0\in\mathcal R_{\mathrm{even}}$ and any sufficiently small $\varepsilon>0$.

\smallskip
\noindent
\textnormal{(iii) If $\lambda_r\neq 0$ and $p$ is odd, then}
\begin{equation}\label{eq:Teps_odd}
T_{\varepsilon,r}
=
\frac{1}{p\alpha_r}
\ln\!\Bigg(
\frac{\varepsilon^{-p}-\operatorname{sign}(y_{r,0})\,\lambda_r/\alpha_r}
{|y_{r,0}|^{-p}-\operatorname{sign}(y_{r,0})\,\lambda_r/\alpha_r}
\Bigg),
\end{equation}
where the logarithm argument is $>1$ for any $x_0\in\mathcal R_{\mathrm{odd}}$ and any sufficiently small $\varepsilon>0$.
\end{proposition}

\begin{proof}
See the Appendix.
\end{proof}

\begin{remark}[Asymptotic rate and dominant modes]\label{rem:Teps_rate}
For $\kappa_r<0$, the asymptotic decay near the origin is exponential and governed by the linear rate $\alpha_r=|\kappa_r|$.
\end{remark}

\subsection{Regions leading to finite-time blow-up}\label{sec:blowup_regions}

Under the shared ODECO basis, each mode of the closed-loop system evolves independently according to~\eqref{eq:mode_bernoulli}; thus finite-time escape of any single mode implies finite-time escape of the full state.

\begin{theorem}[Computable finite-time escape regions]\label{thm:blowup_region_parity}
Consider~\eqref{eq:hom_odeco_closedloop}--\eqref{eq:hom_odeco_decomp} and assume $\kappa_r<0$ for all $r$.
Let $p:=k-2$ and let
$\mathcal I_+,
\mathcal I_-,c_r$ as defined in the Theorem \ref{thm:roa_box}.

\smallskip
\noindent
\textnormal{(i) Even $p$ (equivalently, even $k$).}
If there exists $r\in\mathcal I_+$ such that
$|v_r^\top x_0|>c_r,$
then the corresponding solution escapes to infinity in finite time.

\smallskip
\noindent
\textnormal{(ii) Odd $p$ (equivalently, odd $k$).}
\begin{itemize}
\item If there exists $r\in\mathcal I_+$ such that $v_r^\top x_0>c_r$, then the solution escapes to $+\infty$ in finite time.
Moreover, for such $r\in\mathcal I_+$, any initial condition with $v_r^\top x_0< c_r$ (in particular, any $v_r^\top x_0<0$) yields a forward-complete mode converging to $0$.

\item If there exists $r\in\mathcal I_-$ such that $v_r^\top x_0<c_r$ (note $c_r<0$), then the solution escapes to $-\infty$ in finite time (finite-time blow-down).
\end{itemize}

\smallskip
\noindent
In all cases, the (first) finite-time escape time satisfies
$T_{\mathrm{esc}}
=
\min_{r\in\{1,\dots,n\}}
T_{\mathrm{esc},r},$
where, for a mode $r$ with $\lambda_r\neq 0$ and initial value $y_{r,0}:=v_r^\top x_0$ that satisfies an escape condition,
\begin{equation}\label{eq:escape_time_mode}
T_{\mathrm{esc},r}
=
\frac{1}{p|\kappa_r|}
\ln\!\Bigg(
\frac{\lambda_r/|\kappa_r|}{\lambda_r/|\kappa_r|-y_{r,0}^{-p}}
\Bigg).
\end{equation}
\end{theorem}

\begin{proof}
See Appendix.
\end{proof}

\subsection{Robustness and ISS-type properties under matched disturbances}\label{sec:robust_iss}

We extend the analysis to additive disturbances that are \emph{matched} to the shared ODECO basis for the case of an even $p$. 
Consider
\begin{equation}\label{eq:disturbed_system}
\dot x = Kx + \mathcal A x^{k-1} + d(t),
\qquad
d(t)=\sum_{r=1}^{n} d_r(t)\, v_r,
\end{equation}
where $\mathcal A$ and $K$ satisfy~\eqref{eq:hom_odeco_decomp}.
Let $y(t)=V^\top x(t)$ so that $y_r(t)=v_r^\top x(t)$ and $d_r(t)=v_r^\top d(t)$.
Then each mode satisfies the disturbed scalar dynamics
\begin{equation}\label{eq:disturbed_mode}
\dot y_r = \kappa_r y_r + \lambda_r y_r^{p+1} + d_r(t),
\qquad p:=k-2>0.
\end{equation}

\begin{definition}[\cite{mironchenko2016local}]\label{def:iss}
A system $\dot z = f(z,w)$ is \emph{input-to-state stable (ISS)} with respect to an input $w(\cdot)$ if there exist class-$\mathcal{K}\mathcal{L}$ and class-$\mathcal{K}$ (see \cite{mironchenko2016local}) functions
$\beta$ and $\gamma$ such that for every initial condition and every essentially bounded input,
$|z(t)|
\le
\beta(|z(0)|,t)+\gamma(\|w\|_\infty),
\qquad
\|w\|_\infty:=\mathrm{ess\,sup}_{t\ge 0}|w(t)|.$
If the bound holds only for initial conditions in a neighborhood (or a specified set), we speak of \emph{local ISS} (or ISS on that set).
\end{definition}

We focus on the even-$p$ case (equivalently, even $k$), for which the modal drift admits a convenient symmetry and yields clean robust invariance and ISS-type bounds.

Assume $p$ is even (equivalently, $k$ is even), and that the matched disturbances are essentially bounded:
$|d_r(t)|\le \bar d_r,\qquad \forall t\ge 0,\ \ r=1,\dots,n.$
Let $\mathcal I_+:=\{r:\lambda_r>0\}$.
For each $r\in\mathcal I_+$ define the nominal threshold
$c_{r,*}:=\Bigl(-\frac{\kappa_r}{\lambda_r}\Bigr)^{1/p}$
and $\tilde c_{r,*}$ as the smallest positive solution in $(0,c_{r,*})$ of
\begin{equation}\label{eq:robust_threshold_def_even}
\lambda_r s^{p+1}+\kappa_r s+\bar d_r=0,
\end{equation}
and assume additionally that for each $r\in\mathcal I_+$,
$0<\bar d_r<\bar d_r^{\max}:=\max_{s\in[0,c_{r,*}]} \{-\kappa_r s-\lambda_r s^{p+1}\}$,
so that \eqref{eq:robust_threshold_def_even} admits at least one root in $(0,c_{r,*})$ and hence the smallest one
$\tilde c_{r,*}\in(0,c_{r,*})$ is well-defined.

Define the robust invariant set
\begin{equation}\label{eq:robust_set_even}
\widetilde{\mathcal R}_{\mathrm{even}}
:=
\Bigl\{x_0\in\mathbb R^n:\ 
|v_r^\top x_0|<\tilde c_{r,*}\ \ \forall r\in\mathcal I_+\Bigr\}.
\end{equation}
For each $r\notin\mathcal I_+$ define $\hat c_r$ by: $\hat c_r:=\bar d_r/|\kappa_r|$ if $\lambda_r=0$, and otherwise
$\hat c_r$ is the unique nonnegative solution of
$\lambda_r s^{p+1}+\kappa_r s+\bar d_r=0\ (s\ge 0),$
which exists for all $\lambda_r\le 0$.

\begin{theorem}\label{thm:robust_roa_iss_even}
Consider~\eqref{eq:disturbed_system}--\eqref{eq:disturbed_mode} and assume $\kappa_r<0$ for all $r$.
Then for any $x(0)=x_0\in\widetilde{\mathcal R}_{\mathrm{even}}$, the solution of~\eqref{eq:disturbed_system} exists for all $t\ge 0$ and the set $\widetilde{\mathcal R}_{\mathrm{even}}$ is forward invariant.

Moreover, each mode satisfies an ultimate bound of the form
\begin{equation}\label{eq:ultimate_bound_even}
\limsup_{t\to\infty}|y_r(t)|
\le
\bar c_r,
\qquad
\bar c_r:=
\begin{cases}
\tilde c_{r,*}, & r\in\mathcal I_+,\\[1mm]
\hat c_r, & r\notin\mathcal I_+.
\end{cases}
\end{equation}

Finally, for each mode $r$, the disturbed scalar system~\eqref{eq:disturbed_mode} is ISS on the compact invariant interval
$[-\bar c_r,\bar c_r]$ with an ISS gain that can be chosen as \begin{equation}\label{eq:iss_gain_simple}
\gamma_r(s):=\frac{\bar{c}_r}{\bar{d}_r} s, \quad \gamma_r(\bar d_r)=\bar c_r.
\end{equation}
\end{theorem}

\begin{proof}
Fix a mode $r$ and rewrite~\eqref{eq:disturbed_mode} as
$\dot y_r = F_r(y_r)+d_r(t),
F_r(y):=\kappa_r y+\lambda_r y^{p+1},
|d_r(t)|\le \bar d_r.$
Since $p$ is even, $F_r$ is an odd function: $F_r(-y)=-F_r(y)$.

\emph{Step 1: robust invariance for destabilizing modes ($r\in\mathcal I_+$).}
Let $r\in\mathcal I_+$ so that $\lambda_r>0$ and $\kappa_r<0$.
By definition~\eqref{eq:robust_threshold_def_even}, $F_r(\tilde c_{r,*})+\bar d_r=0$ with $\tilde c_{r,*}\in(0,c_{r,*})$.
At the upper boundary $y_r=\tilde c_{r,*}$, the worst-case pushing outward is $d_r(t)=+\bar d_r$, and we obtain
$\dot y_r \le F_r(\tilde c_{r,*})+\bar d_r=0,$
so trajectories cannot cross $\tilde c_{r,*}$ from below.
By oddness of $F_r$, we also have $F_r(-\tilde c_{r,*})-\bar d_r=0$.
At the lower boundary $y_r=-\tilde c_{r,*}$, the worst-case disturbance pushing outward is $d_r(t)=-\bar d_r$, and
$\dot y_r \ge F_r(-\tilde c_{r,*})-\bar d_r=0.$
Hence $[-\tilde c_{r,*},\tilde c_{r,*}]$ is forward invariant for mode $r$.
Taking the intersection over all $r\in\mathcal I_+$ yields forward invariance of $\widetilde{\mathcal R}_{\mathrm{even}}$.

\emph{Step 2: ultimate bound for the remaining modes ($r\notin\mathcal I_+$).}
For $r\notin\mathcal I_+$ we have $\lambda_r\le 0$ and $\kappa_r<0$.
Since $p$ is even, $y^{p}=|y|^{p}\ge 0$, and therefore
$F_r(y)=\kappa_r y+\lambda_r y|y|^{p}
= y\big(\kappa_r+\lambda_r |y|^{p}\big).$
For any $y>0$, $\kappa_r+\lambda_r |y|^{p}\le \kappa_r<0$, hence $F_r(y)<0$; similarly, for any $y<0$, $F_r(y)>0$.
Thus the drift points toward the origin on both sides, and the disturbance creates a nonzero steady-state offset.

Define $\hat c_r\ge 0$ as the smallest radius such that the inward-pointing conditions hold under worst-case disturbance,
i.e.,
$F_r(\hat c_r)+\bar d_r\le 0,
F_r(-\hat c_r)-\bar d_r\ge 0.$
If $\lambda_r=0$, this yields $\hat c_r=\bar d_r/|\kappa_r|$.
If $\lambda_r<0$, then $\phi_r(s):=\lambda_r s^{p+1}+\kappa_r s+\bar d_r$ satisfies
$\phi_r(0)=\bar d_r>0,\qquad \lim_{s\to\infty}\phi_r(s)=-\infty,$
and $\phi_r'(s)=(p+1)\lambda_r s^{p}+\kappa_r<0$ for all $s\ge 0$, so $\phi_r$ is strictly decreasing and has a unique root $\hat c_r\ge 0$.
By the same inward-pointing argument as Step~1, $[-\hat c_r,\hat c_r]$ is forward invariant and
$\limsup_{t\to\infty}|y_r(t)|\le \hat c_r$.
Collecting all modes yields~\eqref{eq:ultimate_bound_even} and the Euclidean bound on $\|x(t)\|_2$.

\noindent\emph{Step 3: ISS estimate on the invariant set.}
On the compact forward-invariant interval $[-\bar c_r,\bar c_r]$, the scalar drift $F_r$ is continuously differentiable and hence locally Lipschitz.
Moreover, for all $y\in[-\bar c_r,\bar c_r]$ we have
$yF_r(y)=\kappa_r y^2+\lambda_r y^{p+2}=y^2\big(\kappa_r+\lambda_r |y|^p\big)\le -\alpha_r y^2,$
where $\alpha_r:=-(\kappa_r+\max\{\lambda_r,0\}\bar c_r^{\,p})>0$ (note that $\bar c_r< c_{r,*}$ for $r\in\mathcal I_+$ by construction, while $\lambda_r\le 0$ for $r\notin\mathcal I_+$).
Along any solution of $\dot y_r=F_r(y_r)+d_r(t)$, with $|d_r(t)|\le \bar d_r$, define $\rho_r(t):=|y_r(t)|$.
For almost all $t$ with $\rho_r(t)>0$,
$\frac{d}{dt}\frac{1}{2}\rho_r(t)^2 = y_r(t)\dot y_r(t)\le -\alpha_r \rho_r(t)^2+\rho_r(t)\,|d_r(t)|,$
which implies the scalar differential inequality $\dot \rho_r(t)\le -\alpha_r \rho_r(t)+|d_r(t)|$.
By comparison with $\dot z=-\alpha_r z+|d_r(t)|$, we obtain
$|y_r(t)|\le e^{-\alpha_r t}|y_r(0)|+\int_0^t e^{-\alpha_r(t-s)}|d_r(s)|\,ds
\le e^{-\alpha_r t}|y_r(0)|+\frac{1-e^{-\alpha_r t}}{\alpha_r}\,\|d_r\|_\infty.$
Hence (20) is ISS on $[-\bar c_r,\bar c_r]$ in the sense of Definition~1, and in particular one may choose the ISS gain as in \eqref{eq:iss_gain_simple}.
\end{proof}
\begin{remark}
    Our sharp modal decoupling and ROA/time certificates rely on a (shared-basis) ODECO structure.
If the original homogeneous polynomial tensor is not ODECO, one may attempt to
\emph{transform} or \emph{approximate} it by an ODECO tensor via a linear change of coordinates.
A practical pipeline is proposed by Chen~\cite{chen2022explicit}: represent the HPDS by an
``almost-symmetric'' dynamic tensor and fit a \emph{structured CP factorization} under constraints
that enforce an ODECO-type factor relation (implemented, e.g., via constrained nonlinear least squares;
see Algorithm~1 in~\cite{chen2022explicit}).
If the fit error is below a prescribed tolerance, the model is (approximately) ODECO-izable and our results apply; otherwise, the fitted ODECO tensor provides a surrogate model.
In the approximate case, the mismatch between the original tensor term and the ODECO surrogate can be aggregated
as an additive perturbation in the modal coordinates, which motivates applying the robust ISS-type
arguments on the certified set.
\end{remark}

\section{Numerical experiment}

We validate the even-$p$ results using a 2D example with $k=4$ (hence $p=2$).
Let $x=(x_1,x_2)^\top\in\mathbb R^2$ and choose the orthonormal basis
\[
v_1=\begin{bmatrix}\cos\theta\\ \sin\theta\end{bmatrix},\qquad
v_2=\begin{bmatrix}-\sin\theta\\ \cos\theta\end{bmatrix},
\qquad \theta=\frac{\pi}{6},
\]
so that $V=[v_1,v_2]$.
We set
$\kappa_1=\kappa_2=-1,\ \lambda_1=1,\ \lambda_2=-\tfrac12,\ k=4$.
The closed-loop system in the original coordinates is
\begin{equation}\label{eq:exp_even_2d_system}
\begin{split}
\dot x
&=
-(x_1,x_2)^\top
+
\Bigl((v_1^\top x)^3\Bigr)\,v_1
-\frac12\Bigl((v_2^\top x)^3\Bigr)\,v_2,\\
v_1^\top x &=\cos\theta\,x_1+\sin\theta\,x_2,\ \ 
v_2^\top x=-\sin\theta\,x_1+\cos\theta\,x_2.\
\end{split}
\end{equation}

\subsubsection{ROA visualization}\label{sec:exp_even_roa_results}
For even $p$, Theorem~\ref{thm:roa_box} predicts the sharp ROA 
\begin{equation}\label{eq:exp_even_roa_strip}
\mathcal R_{\mathrm{even}}
=
\Bigl\{x_0\in\mathbb R^2:\ |v_1^\top x_0|<c_1\Bigr\},
\quad
c_1=\Bigl(-\frac{\kappa_1}{\lambda_1}\Bigr)^{1/2}=1.
\end{equation}
Figure~\ref{fig:exp_even_roa_basin} shows the basin classification on a grid and the boundary lines $v_1^\top x=\pm 1$, which is indeed in line with Theorem~\ref{thm:roa_box}.

\begin{figure}[t]
  \centering
  \includegraphics[width=0.8\linewidth]{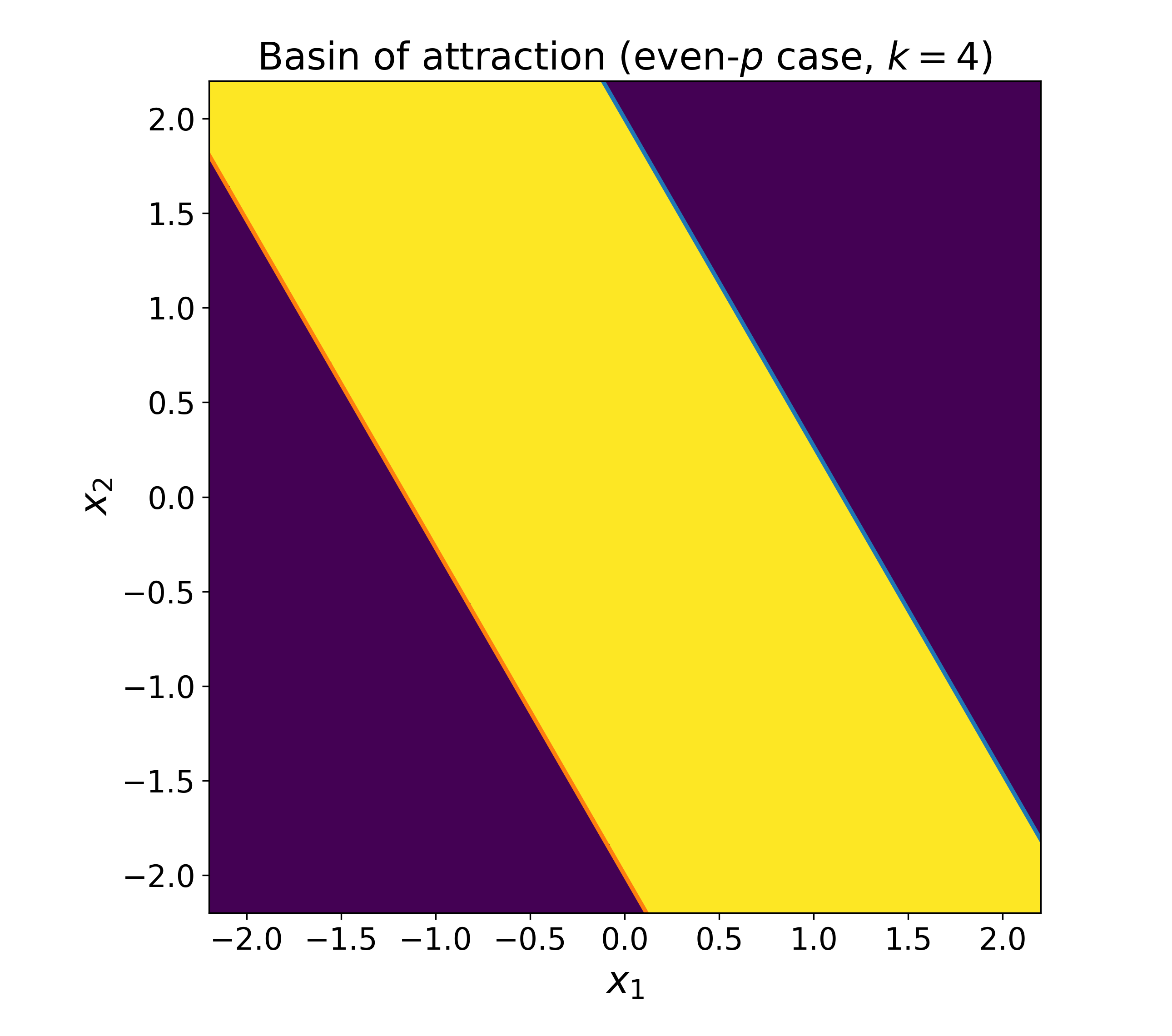}
  \caption{Even-$p$ example ($k=4$): basin classification on a grid.
  The solid lines are the ROA boundaries $v_1^\top x=\pm 1$ in line with that calculated from ~\eqref{eq:exp_even_roa_strip}.}
  \label{fig:exp_even_roa_basin}
\end{figure}

\subsubsection{ISS-type validation under matched disturbances (even $p$)}\label{sec:exp_iss_even}

We consider the matched disturbance model in~\eqref{eq:disturbed_system} with
$d(t)=d_1(t)v_1+d_2(t)v_2$ and choose
$d_1(t)=\bar d\sin(\omega t)$, $d_2(t)=\bar d\cos(0.9\omega t)$ with $\bar d=0.15$.
Figure~\ref{fig:exp_iss_timeseries} shows a representative trajectory in modal coordinates together with the predicted ultimate bounds from Theorem~\ref{thm:robust_roa_iss_even}.
Figure~\ref{fig:exp_iss_gain} sweeps $\bar d$ and compares measured ultimate magnitudes against the predicted curves.

\begin{figure}[t]
  \centering
  \includegraphics[width=0.8\linewidth]{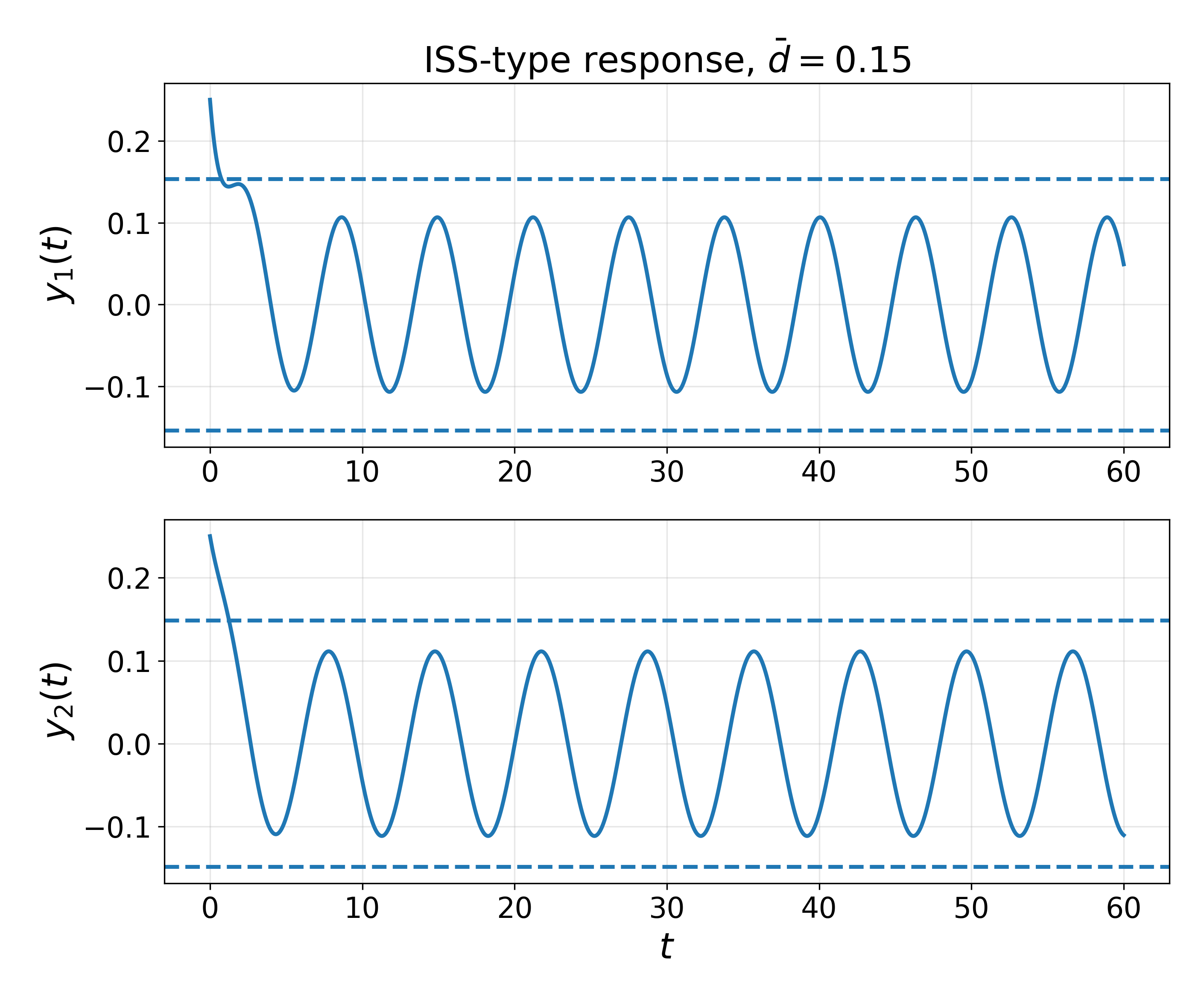}
  \caption{Even-$p$ disturbed case ($\bar d=0.15$): modal time series $y_1(t),y_2(t)$ with the predicted ultimate bounds (dashed).}
  \label{fig:exp_iss_timeseries}
\end{figure}

\begin{figure}[t]
  \centering
  \includegraphics[width=0.95\linewidth]{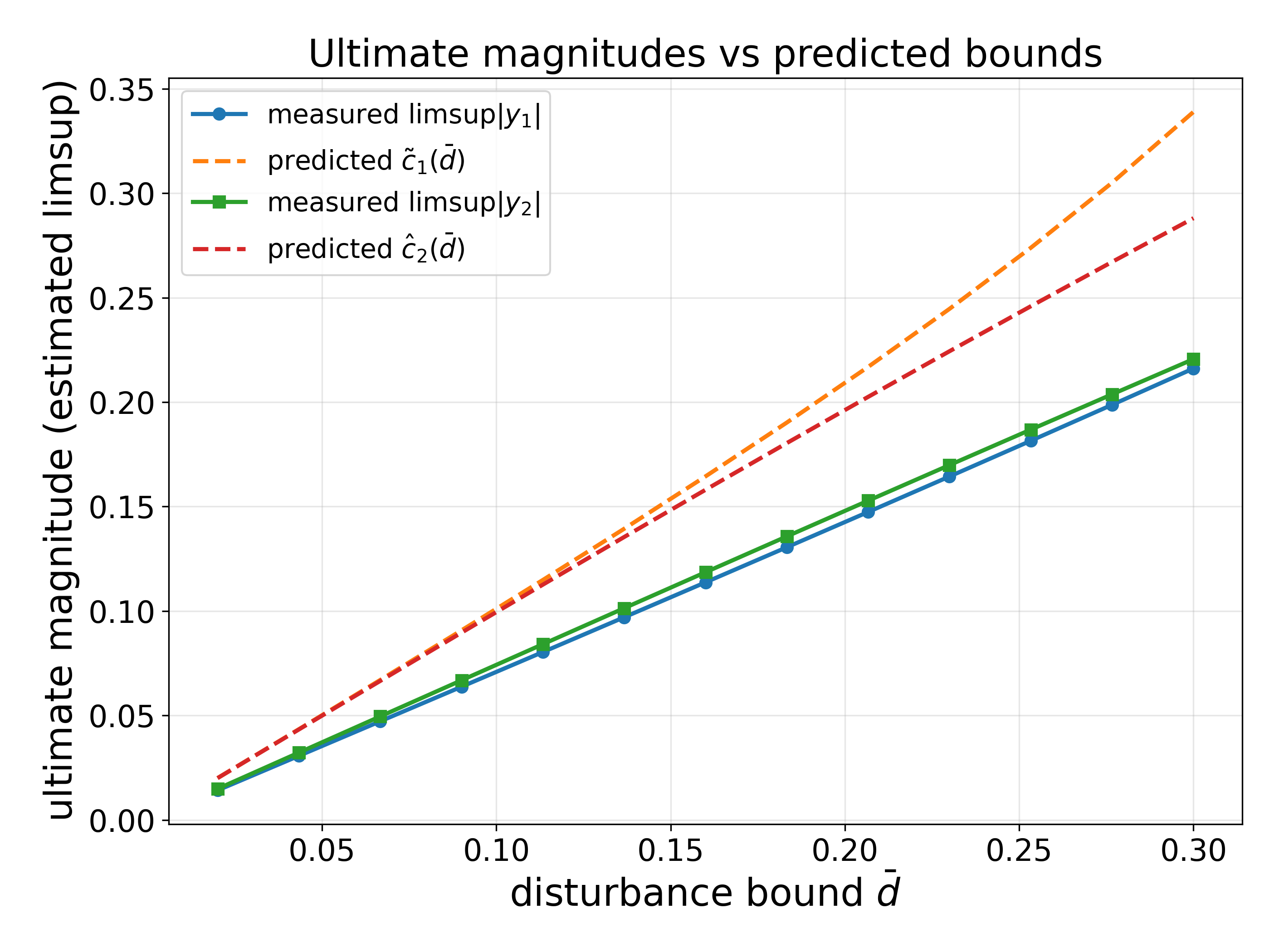}
  \caption{Sweep over $\bar d$: measured ultimate magnitudes (markers) versus predicted bounds (dashed) from Theorem~\ref{thm:robust_roa_iss_even}.}
  \label{fig:exp_iss_gain}
\end{figure}


\section{Conclusion}
This paper proposed a structure-preserving linear feedback design for homogeneous ODECO polynomial systems by enforcing that the controller shares the ODECO eigenbasis of the system tensor. This shared-basis choice yields exact modal decoupling and enables closed-form trajectory expressions, which in turn lead to explicit and sharp region-of-attraction characterizations and computable convergence/escape time estimates. For matched bounded disturbances in the even-$p$ case, we further established ISS-type ultimate bounds. Numerical experiments validated the theoretical findings.

\bibliographystyle{IEEEtran}
\bibliography{bib}

\section{Appendix}

\subsection{Proof of Proposition \ref{prop:settling_time}}
Fix a mode $r$ and consider $\dot y_r=\kappa_r y_r+\lambda_r y_r^{p+1}$ with $\kappa_r<0$.
Along trajectories starting in the ROA, the sign of $y_r(t)$ is preserved and $y_r(t)\to 0$ (Theorem~\ref{thm:roa_box}).
Hence, on the interval $[0,T_{\varepsilon,r})$ we have $|y_r(t)|>\varepsilon>0$ and the transformation below is well-defined.
Let $s_r:=\operatorname{sign}(y_{r,0})\in\{-1,1\}$ and define
$u_r(t):=|y_r(t)|^{-p}>0.$
Using $\frac{d}{dt}|y_r| = s_r \dot y_r$ (since $\operatorname{sign}(y_r(t))=s_r$ on the ROA trajectory), we obtain
\begin{align*}
\dot u_r
&=-p\,|y_r|^{-p-1}\frac{d}{dt}|y_r|
= -p\,|y_r|^{-p-1}\,s_r(\kappa_r y_r+\lambda_r y_r^{p+1})\\
&= -p\Big(\kappa_r |y_r|^{-p} + \lambda_r\, s_r^{\,p+2}\Big)
= -p\kappa_r u_r - p\lambda_r\, s_r^{\,p+2}.
\end{align*}
If $p$ is even, then $p+2$ is even and $s_r^{p+2}=1$, yielding
$\dot u_r = p\alpha_r u_r - p\lambda_r.$
If $p$ is odd, then $p+2$ is odd and $s_r^{p+2}=s_r$, yielding
$\dot u_r = p\alpha_r u_r - p\lambda_r s_r.$
In both cases this is a scalar linear ODE with solution
$u_r(t)=e^{p\alpha_r t}\Big(u_r(0)-\frac{b_r}{\alpha_r}\Big)+\frac{b_r}{\alpha_r},
u_r(0)=|y_{r,0}|^{-p},$
where $b_r=\lambda_r$ if $p$ is even, and $b_r=\lambda_r s_r$ if $p$ is odd.
The hitting condition $|y_r(t)|=\varepsilon$ is equivalent to $u_r(t)=\varepsilon^{-p}$.
Solving for $t$ gives
$T_{\varepsilon,r}
=
\frac{1}{p\alpha_r}\ln\!\Bigg(
\frac{\varepsilon^{-p}-b_r/\alpha_r}{|y_{r,0}|^{-p}-b_r/\alpha_r}
\Bigg),$
which reduces to~\eqref{eq:Teps_even} when $p$ is even and to~\eqref{eq:Teps_odd} when $p$ is odd.
If $\lambda_r=0$, then $\dot y_r=\kappa_r y_r$ and~\eqref{eq:Teps_linear} follows.

 \subsection{Proof of Theorem \ref{thm:blowup_region_parity}}
 Fix a mode $r$ and consider $\dot y=\kappa y+\lambda y^{p+1}$ with $\kappa=-\alpha<0$ and $\lambda\neq 0$.
On any interval where $y(t)\neq 0$, define $u(t):=y(t)^{-p}$.
A direct calculation yields the linear ODE
$\dot u = p\alpha u - p\lambda,$
whose solution is
$u(t)=e^{p\alpha t}\Big(u(0)-\frac{\lambda}{\alpha}\Big)+\frac{\lambda}{\alpha},
 u(0)=y_0^{-p}.$
Finite-time escape occurs if and only if $u(t)$ reaches $0$ at a finite time, since $u=y^{-p}$ implies $|y|\to\infty$ as $u\to 0$.
Solving $u(t)=0$ gives~\eqref{eq:escape_time_mode} whenever it yields a positive time, i.e., whenever
$e^{p\alpha T}=\frac{\lambda/\alpha}{\lambda/\alpha-u(0)}
\quad\Longleftrightarrow\quad
\frac{\lambda/\alpha}{\lambda/\alpha-u(0)}>1.$
This is equivalent to the following:
\[
\lambda>0:\quad 0<u(0)<\frac{\lambda}{\alpha},
\qquad
\lambda<0:\quad \frac{\lambda}{\alpha}<u(0)<0.
\]
We now translate these conditions into explicit regions in terms of $y_0$.

\smallskip
\noindent
\emph{Even $p$.}
Then $u(0)=y_0^{-p}>0$ for any $y_0\neq 0$.
Thus escape can only occur for $\lambda>0$, and $0<u(0)<\lambda/\alpha$ is equivalent to
$|y_0|>\Big(\frac{\alpha}{\lambda}\Big)^{1/p}=c_r,$
which proves (i).

\smallskip
\noindent
\emph{Odd $p$.}
Then $u(0)=y_0^{-p}$ has the same sign as $y_0$.
If $\lambda>0$, the escape condition $0<u(0)<\lambda/\alpha$ requires $u(0)>0$, hence $y_0>0$, and is equivalent to $y_0>c_r$.
If $\lambda<0$, the escape condition $\lambda/\alpha<u(0)<0$ requires $u(0)<0$, hence $y_0<0$, and is equivalent to $y_0<c_r$ (note $c_r<0$).
These give the two one-sided escape regions stated in (ii).
The forward-complete/non-escape statement for the complementary half-lines follows from the sign of $\dot y = y(\kappa+\lambda y^p)$ and the existence of the nonzero equilibrium at $y=c_r$.

Finally, if any mode escapes in finite time, then $\|x(t)\|_2=\|y(t)\|_2\ge |y_r(t)|\to\infty$, hence the full state escapes in finite time.

\end{document}